\documentclass[11pt]{article}
\usepackage[a4paper,margin=1in]{geometry}
\usepackage[T1]{fontenc}
\usepackage[utf8]{inputenc}
\usepackage{lmodern}
\usepackage[expansion=false]{microtype}
\usepackage{amsmath,amssymb,amsfonts,amsthm,mathtools}
\usepackage{booktabs}
\usepackage{array}
\usepackage{enumitem}
\usepackage{graphicx}
\usepackage{xcolor}
\usepackage{hyperref}
\usepackage{url}
\usepackage{csquotes}
\usepackage[capitalise,nameinlink]{cleveref}
\crefname{assumption}{Assumption}{Assumptions}
\Crefname{assumption}{Assumption}{Assumptions}
\crefname{claim}{Claim}{Claims}
\Crefname{claim}{Claim}{Claims}
\usepackage{tikz}
\usepackage{pgfplots}
\pgfplotsset{compat=1.18}
\usepackage[backend=biber,style=authoryear,maxcitenames=2,maxbibnames=99,giveninits=true,uniquename=init,doi=true,url=true]{biblatex}
\hypersetup{
  colorlinks=true,
  linkcolor=blue!55!black,
  citecolor=blue!55!black,
  urlcolor=blue!55!black,
  pdftitle={Financial Epiplexity: A Theory of Learnable Market Structure under Bounded Computation},
  pdfauthor={Miquel Noguer i Alonso}
}

\theoremstyle{definition}
\newtheorem{definition}{Definition}[section]
\newtheorem{assumption}{Assumption}[section]
\newtheorem{example}{Example}[section]

\theoremstyle{plain}
\newtheorem{proposition}{Proposition}[section]
\newtheorem{theorem}{Theorem}[section]

\newtheorem{corollary}{Corollary}[section]
\newtheorem{claim}{Claim}[section]

\theoremstyle{remark}
\newtheorem{remark}{Remark}[section]
\newtheorem{conjecture}{Conjecture}[section]

\DeclareMathOperator{\E}{\mathbb{E}}
\DeclareMathOperator{\Var}{Var}
\DeclareMathOperator{\Cov}{Cov}

\DeclareMathOperator{\MDL}{MDL}

\newcommand{\F}{\mathcal{F}}
\newcommand{\M}{\mathcal{M}}

\newcommand{\X}{\mathcal{X}}

\newcommand{\Prob}{\mathbb{P}}
\newcommand{\FEpi}{\operatorname{FEpi}}
\newcommand{\TBE}{\operatorname{TBE}}
\newcommand{\FMDL}{\operatorname{FMDL}}

\newcommand{\OO}{\mathcal{O}}

\newcommand{\eps}{\varepsilon}
\newcommand{\Gain}{G}
\newcommand{\MDLzero}{\operatorname{MDL}_0}
\newcommand{\SR}{\operatorname{SR}}
\newcommand{\IC}{\operatorname{IC}}
\newcommand{\IR}{\operatorname{IR}}

\newcommand{\TV}{\operatorname{TV}}

\title{\vspace{-1.2cm}\textbf{Financial Epiplexity:\\ A Theory of Learnable Market Structure under Bounded Computation}}
\author{Miquel Noguer i Alonso\\ \\Artificial Intelligence Finance Institute}
\date{\today}

\begin{document}
\maketitle

\begin{abstract}
\noindent
Financial markets are hard to predict, not because every price move is ontically random, but because structure is strategic, capacity constrained, and computationally difficult. Classical financial information theory measures uncertainty, distributional change, dependence, and directed information flow through entropy, KL divergence, NMI, and transfer entropy. This paper keeps that foundation and asks a theoretical question: how much detected structure can a bounded investor learn and reuse? We develop \emph{financial epiplexity}: a time-bounded MDL measure of learnable market structure relative to a filtration, representation, target, model class, and budget. The alpha-relevant object is not raw complexity, but target-specific net compression gain beyond a benchmark after charging the model and representation for their description length. This framing is consistent with no-arbitrage and the \(P\)--\(Q\) wedge: risk-neutral martingality closes arbitrage under a pricing measure, while epiplexity concerns real-world learnability under an information set. We prove separation results showing equal entropy need not imply equal epiplexity, derive finite-sample thresholds for useful regimes, formalize representation dependence and nonmonotonicity in compute, and state when additional variables become structurally valuable. Using the Kelly--Cover--Barron--Cover code-length/wealth correspondence, we bound cumulative excess log-growth, Sharpe ratios, sustainable information coefficients, and Grinold-style breadth by structural bits per period, with leverage and survival treated as growth/capacity constraints rather than Sharpe-ratio deflators. The dynamic theory models alpha decay as migration of private bits into market budget, crowding as mutual compressibility, and capacity as bit leakage through trading. The strategic layer studies budget choice, signal congestion, competitive revelation, endogenous computational depth, heterogeneity, and Red Queen compute competition.
\end{abstract}


\newpage
\tableofcontents
\newpage

\section{Introduction}

A financial market is not ontically random in the sense of a physical chance device; it is hard to predict because causal structure is filtered through information, orders, balance sheets, constraints, institutions, and strategic feedback. It is also a computational object observed through prices, volumes, order books, macroeconomic releases, central-bank language, credit spreads, volatility surfaces, news, regulations, flows, and portfolio constraints. Yet most quantitative measures of information used in finance treat the data primarily as uncertainty, dependence, or entropy. Volatility measures dispersion. Entropy measures average surprise. Mutual information measures statistical dependence. Predictive accuracy measures performance for a particular target under a particular model. None of these quantities directly asks the question that matters to a bounded financial learner:
\begin{quote}
\emph{How much useful market structure can this data teach a finite model under finite computation?}
\end{quote}
This question is central to modern finance because most financial learning systems are bounded in at least four ways: they have finite samples, finite compute, finite memory, and finite time before the regime changes. In markets, the observer is never the omniscient statistician of classical asymptotics. It is a fund, bank, agent, risk desk, or algorithmic system operating under latency, capital, compliance, and transaction-cost constraints.

This positioning is consistent with the companion argument that markets are not literally random but hard to predict: no-arbitrage, informational efficiency, learnability, and net exploitability are distinct notions; the risk-neutral measure is an instrumental pricing measure, not the physical data-generating law; and positive prediction is economically relevant only when it survives costs, capacity, and survival constraints \parencite{noguer2026hardtopredict}. The present paper adds a bounded-computation layer to that thesis. It asks not whether a pattern exists in principle, but how many target-specific structural bits a feasible learner can extract from a stated information set.

The concept of \emph{epiplexity}, introduced by \textcite{finzi2026epiplexity}, provides a language for this problem. Epiplexity is designed to capture structural information available to a computationally bounded observer, separating learnable structure from residual time-bounded entropy. The original motivation comes from tensions between classical information theory and modern AI practice: deterministic synthetic data can teach models useful behavior; the order of data can matter; likelihood modeling can produce models that appear to learn richer structures than those explicitly present in the data-generating procedure. In the epiplexity framework, these phenomena are not paradoxes. They are consequences of bounded computation and representation-dependent access to structure.

This paper argues that finance is one of the natural domains for epiplexity. Financial markets contain persistent but unstable regularities: volatility clustering, leverage effects, factor structure, liquidity spirals, credit cycles, calendar effects, behavioral feedback, macro transmission, regime switches, and crisis dynamics. They also contain enormous residual randomness. The usual statement that ``markets are noisy'' is correct but incomplete. A more precise statement is:
\begin{quote}
\emph{Financial data contain learnable structure, but the visible amount of structure depends on representation, computational budget, horizon, task, and market regime.}
\end{quote}
Financial epiplexity is proposed as a formal measure of that visible, learnable structure.

\subsection{The central distinction}

The core distinction is:
\[
\text{financial entropy} \neq \text{financial epiplexity}.
\]
A purely random return sequence can have high entropy but low epiplexity because it teaches little reusable structure. A trivial calendar rule can be easy to learn but low in epiplexity because it carries little structural richness. A multimodal dataset combining prices, volatility, liquidity, rates, credit, macro surprises, options, text, and forward outcomes may be harder to learn, yet higher in epiplexity because it encodes reusable market mechanisms.

In one line:
\[
\boxed{\text{Financial epiplexity is bounded-compute learnable market structure.}}
\]
This definition is intentionally relative. It is not a universal scalar attached to a market once and for all. It depends on an observer class, a compute budget, a representation of data, a predictive task, and a horizon. This relativity is a strength, not a weakness. Markets are precisely systems in which structure appears differently to different observers.

\subsection{Contributions}

This theory paper makes eight contributions.

\begin{enumerate}[label=(\roman*)]
\item It defines \emph{financial epiplexity} through a time-bounded MDL decomposition of represented financial data into learned model bits and residual predictive bits.
\item It separates financial epiplexity from entropy, volatility, mutual information, in-sample fit, and ordinary learnability.
\item It proves elementary separation results: equal entropy need not imply equal financial epiplexity; representation and temporal ordering can change accessible structure; epiplexity need not be monotone in compute; and memorization is not useful structure.
\item It introduces the alpha-relevant object: net, target-specific MDL gain, especially return-targeted conditional epiplexity $\mathcal A_B(Y;Z\mid X)$.
\item It proves a monetization theorem: no bounded strategy can extract more lifetime excess log-growth than a market-structure constant times the target-specific compression gain of its data and representation.
\item It derives a Sharpe ceiling, a sustainable information-coefficient ceiling, a computational-depth law for alpha decay, a mutual-compressibility view of crowding, and a bit-leak interpretation of capacity.
\item It formalizes the importance of one-way and inferential orderings in market data, especially when downstream transfer requires latent mechanism recovery.
\item It adds the strategic layer: a budget game, a congestion game, a Kyle-style competitive-revelation game, an endogenous depth game, and a Red Queen arms-race game for aggregate computation.
\end{enumerate}

\subsection{A guiding example}

Consider three datasets for predicting a one-month equity drawdown:
\begin{align*}
D_1 &= \{\text{daily returns only}\},\\
D_2 &= \{\text{returns, realized volatility, volume, sector factors}\},\\
D_3 &= \{\text{returns, volatility, rates, credit spreads, option skew,}\nonumber\\
&\qquad \text{macro surprises, news embeddings}\}.
\end{align*}
The first dataset may be noisy and high entropy. The second may expose volatility clustering and factor structure. The third may encode macro-financial mechanisms: central-bank surprises, credit tightening, volatility demand, liquidity contraction, and equity factor rotation. If a bounded learner trained on $D_3$ learns reusable subprograms that transfer to crisis periods, then $D_3$ has higher financial epiplexity than $D_1$, even if its training loss is initially higher.

This example also shows why financial epiplexity is not simply predictability. A dataset can be predictable because it is trivial, or because it contains deep structure. Epiplexity targets the second case.

\section{Classical Information and the Financial Learning Problem}

\subsection{Shannon information}

For a random variable $X$ with probability mass function $p$, the Shannon information of outcome $x$ is
\[
I(x) = -\log_2 p(x),
\]
and the Shannon entropy is
\[
H(X)=\E[-\log_2 p(X)].
\]
Entropy measures average surprise, not usefulness. A market path with independent random signs can have maximal entropy and yet teach no stable trading rule. Conversely, a regime process with the same one-step marginal entropy can be highly useful if its temporal dependence is learnable.

Mutual information,
\[
I(X;Y)=H(Y)-H(Y\mid X),
\]
measures statistical dependence between variables. It is valuable, but it is still not sufficient for the problem of financial learning. It does not specify whether the dependence is accessible to a bounded learner, whether it is stable across regimes, whether it survives transaction costs, or whether it can be represented in a model class under a compute budget.

\begin{remark}[Consistency with classical financial information theory]
This paper is intended as a theoretical continuation of, not a replacement for, classical financial information-theoretic diagnostics. In the notation of \textcite{noguer2025financialinformation}, entropy measures return uncertainty, KL divergence measures distributional regime change, normalized mutual information provides a bounded diagnostic of temporal dependence and market efficiency, and transfer entropy measures directional information flow. Financial epiplexity takes these objects as first-layer diagnostics and asks a second-layer question: conditional on a task, representation, and budget, which part of the detected uncertainty reduction or dependence can be compressed into a reusable bounded model? Thus entropy, KL, NMI, and transfer entropy diagnose uncertainty, change, dependence, and direction; financial epiplexity asks whether those diagnostics become target-relevant predictive code. Throughout the MDL formulas below, logarithms are base two and code lengths are measured in bits. When a wealth or utility expression is written in natural-log units, the factor $\ln 2$ converts bits to nats.
\end{remark}

\subsection{Algorithmic information}

Kolmogorov complexity $K(x)$ is the length of the shortest program that outputs $x$ and halts \parencite{kolmogorov1965,li2008kolmogorov}. It applies to individual strings rather than only random variables. However, it is incomputable and assumes unbounded search over programs. A financial string may have a short generative program that is practically impossible to discover before the regime disappears. For finance, the relevant question is not only whether a short program exists, but whether a bounded learner can find useful structure in time. The raw model-bit notion is closest to budget-limited sophistication or structure-function thinking in algorithmic information theory \parencite{gacs2001,vereshchagin2004}; the term computational depth also has a prior literature, from Bennett's logical depth to later complexity-theoretic versions \parencite{bennett1988,antunes2006}. Here depth means a market-specific budget gap in extractable predictive structure, not identical to either classical notion.

\subsection{Minimum description length}

The minimum description length principle selects models that compress data well by trading off model complexity and residual error \parencite{rissanen1978,barron1998,grunwald2007}. In its simplest two-part form,
\[
\MDL(D;M)=L(M)+L(D\mid M),
\]
where $L(M)$ is the number of bits required to describe the model and $L(D\mid M)$ is the number of bits required to describe the data given the model.

Time-bounded description-length ideas also connect to the speed prior and to PAC-Bayes coding views of learning \parencite{schmidhuber2002,mcallester1999,catoni2007}. Epiplexity is closely related to this two-part view, but it changes the interpretation. The model bits are not merely a penalty for complexity; they are treated as the structural information that the data have taught the model, provided the model is selected under a bounded computational search.

\subsection{Why finance needs bounded observers}

The efficient market hypothesis, in its classical form, says that prices reflect available information \parencite{fama1970}. Grossman and Stiglitz famously argued that perfectly informationally efficient markets are impossible when information is costly, because if prices fully reveal information there is no incentive to acquire it \parencite{grossman1980}. Lo's adaptive markets hypothesis reframes market efficiency through adaptation, competition, and changing environments \parencite{lo2004}.

Financial epiplexity is compatible with this line of thought. It does not say that markets are inefficient in a universal sense. It says that market structure is observer-relative. The same data may be noise for one learner, factor structure for another, and regime mechanism for a third. The relevant object is therefore not ``information in the market'' in the abstract, but accessible structural information under computational and institutional constraints.

This also locates the paper relative to the anomaly and financial econometrics and financial-machine-learning literatures \parencite{campbell1997,lopezdeprado2018}. Those literatures ask whether a published predictor earns returns, whether the claim survives multiple testing, and whether the effect decays after dissemination \parencite{harvey2016,mclean2016,bailey2014,lopezdeprado2018}. Financial epiplexity asks a prior theoretical question: before monetization and before multiple testing, how many target-specific structural bits are accessible to a bounded learner from a represented data source? The answer need not equal realized alpha. It is the information budget from which any alpha, risk improvement, or capacity claim must be paid.

A related tradition models bounded observers through information-capacity constraints: rational inattention posits agents who optimally allocate a limited flow of mutual information between states and actions \parencite{sims2003}. Financial epiplexity is complementary but distinct. In rational inattention the constraint is attentional: the agent chooses \emph{which} bits to observe, and the bits themselves are classical Shannon bits, available at a price. Here the constraint is computational: the bits are defined only relative to a model class and budget, and structure can be present in fully observed data yet inaccessible because extracting it requires computation the observer does not have. An inattentive agent with unbounded computation and a full-capacity channel recovers the classical benchmark; a bounded learner does not. In the same spirit, computational reformulations of market efficiency tie the presence of exploitable patterns to computational hardness \parencite{maymin2011,hasanhodzic2011}. Epiplexity gives that intuition a quantitative unit: bits of structure per budget.

\section{Epiplexity: The Conceptual Source}

\textcite{finzi2026epiplexity} define epiplexity as a measure of structural information available to a computationally bounded observer. The original framework separates structural content from residual time-bounded entropy and is motivated by three tensions:
\begin{enumerate}[label=(\alph*)]
\item deterministic transformations appear not to increase classical information, yet synthetic data and self-play can teach useful behavior;
\item classical information is invariant to factorization order, yet learners depend strongly on ordering and representation;
\item likelihood modeling is often described as distribution matching, yet trained models can acquire reusable subprograms and emergent capabilities.
\end{enumerate}

The finance translation is direct. A deterministic stress simulator can teach a risk model about crisis dynamics because the learner is not charged with knowing the simulator's source code; it observes trajectories and must discover a reusable description under its own budget. If the simulator code is handed to the learner as public information, the epiplexity of the generated paths conditional on that code can collapse. If only the paths are observed, the model may need to learn intermediate crisis mechanisms: margin spirals, volatility feedback, credit widening, liquidity withdrawal, and liquidation cascades. The useful uncertainty is therefore epistemic rather than aleatory. The same distinction applies to real markets: the generative mechanism may be partly deterministic, but the bounded investor faces a discovery problem. A time-ordered sequence of central-bank statements and yield-curve reactions is more useful than a shuffled bag of sentences and returns. A likelihood model trained on financial text may learn macroeconomic and institutional relationships that transfer to portfolio decisions.

\subsection{Epiplexity is not ordinary learnability}

Epiplexity should not be confused with ease of learning. Let $D$ be a dataset.
\begin{itemize}
\item If $D$ is pure white noise, it is hard to learn and low in epiplexity.
\item If $D$ is a trivial repetition, it is easy to learn and low in epiplexity.
\item If $D$ contains rich, reusable structure, it may be hard to learn but high in epiplexity.
\end{itemize}
For finance, this distinction is essential. A backtest that learns a trivial seasonal artifact need not have high epiplexity. A multimodal dataset that gradually teaches a model how liquidity, leverage, volatility, and credit interact may have high epiplexity even if its first-stage loss declines slowly.

\section{Financial Epiplexity: Formal Definition}

\subsection{Financial data as filtered processes}

Let $(\Omega,\F,(\F_t)_{t\geq 0},\Prob)$ be a filtered probability space. Let $P_t$ denote prices, $r_t=\log P_t-\log P_{t-1}$ returns, and $Z_t$ a possibly latent market state. Let $\mathcal I_t$ denote the raw information available at time $t$ to a financial observer: prices, volumes, order-book features, macro releases, news, balance-sheet data, option surfaces, flows, and internal agent logs.

A representation map
\[
R: \mathcal I_t \mapsto X_t^R \in \X
\]
turns raw market information into features. A target map defines
\[
Y_{t+h}=\tau(\mathcal I_{t+1:t+h}),
\]
where $Y_{t+h}$ may be a return, sign, drawdown, realized volatility, VaR breach, regime label, liquidity shock, or portfolio loss.

The represented financial dataset is
\[
D_{T,h}^R = \{(X_t^R,Y_{t+h})\}_{t=1}^{T-h}.
\]
The superscript $R$ is important. Financial epiplexity is not a property of raw data alone; it is a property of a representation of data for a task.

The filtration is equally important. A discounted price can be close to a martingale relative to a public filtration and still contain structure relative to an enlarged private filtration, a faster microstructure filtration, or a richer model-implied state. Hence every epiplexity statement below is implicitly indexed by an information set: it is a claim about what is learnable from \(\mathcal I_t\) as represented by \(R\), not a metaphysical claim about the asset price process itself.

\subsection{Bounded model classes}

Let the computational budget be a vector
\[
\mathbf B=(B_{\mathrm{train}},B_{\mathrm{eval}},B_{\mathrm{mem}},B_{\mathrm{search}},B_{\mathrm{lat}},B_{\mathrm{cost}}),
\]
covering training FLOPs, evaluation latency, memory, hyperparameter search, and institutional cost. A scalar symbol $B$ below denotes a fixed feasible set induced by such a vector, not a universal unit of computation. Formally,
\[
\M_B=\{M: C(M;D,R)\preceq \mathbf B\},
\]
where $\preceq$ is the componentwise budget order. The model class is fixed before inspecting the target sample; otherwise the code can hide data mining inside the model-class choice.

Each model $M\in\M_B$ defines a conditional predictive distribution
\[
p_M(y\mid x).
\]
The predictive code length in bits is
\[
\ell_M(y\mid x) = -\log_2 p_M(y\mid x).
\]
The description length charged to the learner is
\[
L_R(M)=L(R)+L(M\mid R),
\]
including the representation map, preprocessing rules, architecture, parameters, state variables, and any non-public simulator or feature-construction code. When the representation is fixed by the paper and not searched, $L(R)$ is a constant and can be suppressed. When representations are selected from a family, $L(R)$ is part of the MDL penalty; this is the formal guard against representation mining.

\begin{definition}[Time-bounded financial MDL]
For a represented financial dataset $D_{T,h}^R$, model class $\M_B$, and code length $L_R$, define
\[
\FMDL_B(D_{T,h}^R)
=
\inf_{M\in\M_B}
\left\{
L_R(M)+\sum_{t=1}^{T-h} -\log_2 p_M(Y_{t+h}\mid X_t^R)
\right\}.
\]
The infimum need not be attained and, even when it is, two-part codes can have multiple near-minimizers with different model-bit/data-bit splits. Fix a tolerance $\eps_T\ge0$ and define the $\eps_T$-optimal set
\[
\mathcal S_{B,\eps_T}(D)=
\left\{M\in\M_B:
L_R(M)+\sum_{t=1}^{T-h}\ell_M(Y_{t+h}\mid X_t^R)
\le \FMDL_B(D_{T,h}^R)+\eps_T
\right\}.
\]
Let $M_{B,\eps}^{\star}$ be an element of $\mathcal S_{B,\eps_T}(D)$ with minimal $L_R(M)$, with a fixed lexicographic tie-breaker if needed. This tie-breaking convention is part of the definition below and avoids the usual sophistication-instability ambiguity.
\end{definition}

\begin{definition}[Raw financial epiplexity]
The raw financial epiplexity of $D_{T,h}^R$ under budget $B$ and tolerance $\eps_T$ is
\[
\FEpi_B^{\mathrm{raw}}(D_{T,h}^R)=L_R(M_{B,\eps}^{\star}).
\]
The financial time-bounded entropy is
\[
\TBE_B^{\mathrm{fin}}(D_{T,h}^R)=
\sum_{t=1}^{T-h} -\log_2 p_{M_{B,\eps}^{\star}}(Y_{t+h}\mid X_t^R).
\]
Thus, up to the chosen tolerance,
\[
\FMDL_B(D_{T,h}^R)=\FEpi_B^{\mathrm{raw}}(D_{T,h}^R)+\TBE_B^{\mathrm{fin}}(D_{T,h}^R)+\OO(\eps_T).
\]
The paper's alpha bounds use the stable object $\FMDL_B$ and its associated gain; raw epiplexity is retained as an expository decomposition rather than as the monetizable quantity.
\end{definition}

The raw definition follows the conceptual epiplexity decomposition: model bits represent learned structure, residual bits represent what remains unpredictable to the bounded learner. In overparameterized neural networks, the literal uncompressed parameter count is not the intended $L_R(M)$; the relevant object is a code length for the trained predictor, such as compressed weights, stochastic complexity, PAC-Bayes code length, or another fixed compression scheme. This prevents billions of unused parameters from mechanically making every dataset high-epiplexity.

\subsection{Benchmark-adjusted financial epiplexity}

In finance, a benchmark adjustment is useful. Let $M_0$ be a null model, such as an iid return model, historical volatility model, random-walk model, GARCH benchmark, or no-skill probability model. Define
\[
\MDL_0(D)=L_R(M_0)+\sum_t \ell_{M_0}(Y_{t+h}\mid X_t^R).
\]

\begin{definition}[Incremental financial epiplexity]
The incremental financial epiplexity relative to $M_0$ is
\[
\Delta\FEpi_B(D_{T,h}^R;M_0)=L_R(M_{B,\eps}^{\star})-L_R(M_0),
\]
and the net MDL gain is
\[
G_B(D_{T,h}^R;M_0)=\MDL_0(D)-\FMDL_B(D).
\]
A representation has \emph{positive net financial epiplexity} relative to $M_0$ if
\[
G_B(D_{T,h}^R;M_0)>0.
\]
\end{definition}

The distinction between $\Delta\FEpi_B$ and $G_B$ is crucial. A more complex model can have higher model bits but still be useless if it does not reduce residual predictive bits enough. In financial terms, complexity is justified only when it buys robust compression of future-relevant uncertainty.

Thus the monetizable quantity is not raw model size, but the total compression improvement over the null after charging the model for its own description length. We write \(\Gain_B\) for this gain when no ambiguity is possible.

\begin{definition}[Return-targeted conditional epiplexity]
Let $X$ denote public information, $Z$ proprietary or alternative information, and $Y$ the future return, loss, drawdown, volatility, or breach target. The return-targeted conditional epiplexity at budget $B$ is
\[
\mathcal A_B(Y;Z\mid X)
=
\left[
\FMDL_B(Y\mid X)-\FMDL_B(Y\mid X,Z)
\right]_+ .
\]
This is the bounded-compute predictive compression of the financial target supplied by $Z$ beyond $X$. It is the clean finance object behind the phrase ``useful market structure.''
\end{definition}

\subsection{Time-varying epiplexity}

Financial structure is local in time. For a rolling window $W_t=\{t-w+1,\dots,t\}$, define
\[
\FEpi_{B,t}^{\mathrm{raw}}(w;R,h)=L_R(M_{B,t,w}^\star),
\]
where $M_{B,t,w}^\star$ is the corresponding minimal-description, $\eps_w$-optimal code on $D_{t-w+1:t,h}^R$. The associated rolling gain is
\[
\Gain_{B,t}(w;M_0)=\MDLzero(D_{t-w+1:t,h}^R)-\FMDL_B(D_{t-w+1:t,h}^R).
\]
A regime change is then not an exception to the theory; it is a change in the local epiplexity profile. This definition also makes explicit that finite-sample estimates should be reported with a window, horizon, representation, and budget.

\subsection{Epiplexity ratio}

A normalized diagnostic is
\[
\rho_B(D)=\frac{\FEpi_B^{\mathrm{raw}}(D)}{\FEpi_B^{\mathrm{raw}}(D)+\TBE_B^{\mathrm{fin}}(D)}.
\]
A very low $\rho_B$ indicates that most bits remain residual entropy. A very high $\rho_B$ can indicate either genuine structure or overfit/memorization; therefore $\rho_B$ must be interpreted together with out-of-sample loss and model stability.

\section{Interpretation in Finance}

\subsection{What financial epiplexity measures}

Financial epiplexity measures the model bits needed to encode market structure that a bounded learner can use. Examples include:
\begin{itemize}
\item latent regimes and transition probabilities;
\item volatility clustering and leverage effects;
\item factor loadings and factor rotations;
\item yield-curve and credit-spread transmission;
\item option-implied skew and volatility risk premia;
\item liquidity spirals and transaction-cost states;
\item institutional language in central-bank and earnings-call text;
\item synthetic crisis mechanisms that transfer to stress testing;
\item causal or quasi-causal market narratives stored in agent logs.
\end{itemize}

Financial time-bounded entropy is the residual uncertainty after the best bounded model has extracted what it can. This residual includes genuine randomness, unobserved information, adverse selection, microstructure noise, nonstationarity, and structure that exists but is inaccessible under the budget.

\subsection{Relation to stylized facts}

Financial returns exhibit stylized facts: heavy tails, volatility clustering, weak linear autocorrelation of returns, nonlinear dependence, aggregational Gaussianity, and leverage effects \parencite{mandelbrot1963,cont2001}. A stylized fact is not automatically epiplexity. It becomes epiplexity only when a bounded model can encode it in reusable form and improve predictive coding, risk estimation, or transfer. For example, volatility clustering has high financial epiplexity for risk management if a GARCH, stochastic-volatility, Markov-switching, or neural volatility model reduces log-loss for future volatility or drawdown \parencite{engle1982,bollerslev1986,hamilton1989}. The same fact may have lower epiplexity for next-day directional equity prediction.

\subsection{Relation to market efficiency}

Financial epiplexity does not refute market efficiency. It refines the observer. A market may be efficient relative to simple price-only strategies and still contain high epiplexity for a learner with access to richer representations, longer memory, alternative data, or better compute. Conversely, a market may show in-sample predictability that has low epiplexity because it is fragile, nontransferable, or transaction-cost dominated.

Thus the efficient-market question should be decomposed as:
\[
\text{efficiency relative to what representation, compute budget, horizon, and cost structure?}
\]

\subsection{Relation to martingales, the \texorpdfstring{\(P\)--\(Q\)}{P-Q} wedge, and net exploitability}

The framework also remains consistent with the standard apparatus of mathematical finance. No-arbitrage supplies an equivalent pricing measure \(Q\) under which discounted prices are local martingales; it does not imply that the physical conditional law under \(P\) contains no structure, nor that prices are independent Gaussian random walks \parencite{noguer2026hardtopredict}. Epiplexity is therefore not a violation of martingality. It is a statement about bounded predictive compression under a specified real-world information set and target.

This distinction is clearest through a Doob-style decomposition. For an adapted cumulative return or payoff process, write informally
\[
X_n=X_0+M_n+A_n,
\]
where \(M\) is a martingale innovation relative to the filtration and \(A\) is the predictable component. Financial epiplexity can attach to \(A\), to conditional variance, to jump intensity, to tail shape, or to latent regime structure. The remaining innovation, together with inaccessible structure, is financial time-bounded entropy. Thus the theory does not say that all predictability is alpha: predictable premia, volatility states, and liquidity conditions may be real yet risk-compensated, already priced, or costly to trade.

Similarly, the \(P\)--\(Q\) wedge belongs to the first-layer information ledger. In the lognormal benchmark its relative entropy is proportional to one half of squared Sharpe times horizon; this measures the informational distortion needed to price risk, not by itself a usable trading edge. In the present theory, the wedge can enter the null benchmark \(M_0\), the target definition, or the monetization constant \(\kappa\), while the epiplexity term measures additional bounded-compute compression under the physical learning problem. Net exploitability is obtained only after risk adjustment, costs, impact, capacity, and survival.

\section{Basic Results}

\subsection{Equal entropy does not imply equal financial epiplexity}

\begin{proposition}[Entropy-epiplexity separation]
There exist two stationary binary return-sign processes $X^{(0)}$ and $X^{(1)}$ such that
\[
H(X_t^{(0)})=H(X_t^{(1)})=1
\]
for every $t$, but for sufficiently large $T$ and a model class containing iid Bernoulli and first-order Markov models,
\[
\FEpi_B^{\mathrm{raw}}(D_T^{(1)})>\FEpi_B^{\mathrm{raw}}(D_T^{(0)})
\]
and
\[
\FMDL_B(D_T^{(1)})<\FMDL_B^{\mathrm{iid}}(D_T^{(1)}),
\]
whereas $D_T^{(0)}$ has no positive net gain over the iid model.
\end{proposition}

\begin{proof}
Let $X_t^{(0)}$ be iid Bernoulli$(1/2)$. Let $X_t^{(1)}$ be a two-state Markov chain on $\{0,1\}$ with stationary distribution $(1/2,1/2)$ and transition probability $q<1/2$ of switching state. Both processes have one-period marginal entropy equal to one bit. For $X^{(0)}$, the iid Bernoulli model is optimal up to finite-sample fluctuations and gives residual code length approximately $T$ bits with constant model description length. For $X^{(1)}$, a first-order Markov model has conditional entropy $h_2(q)<1$, where $h_2$ is the binary entropy function, and residual code length approximately $T h_2(q)$ plus $\OO(\log T)$ parameter-coding bits. Once $T(1-h_2(q))$ exceeds the additional description length of the Markov transition parameters, the Markov model achieves shorter MDL. Its model bits encode temporal structure absent in the iid process. Hence equal marginal entropy does not imply equal financial epiplexity.\end{proof}

\begin{remark}
In finance, this proposition says that a random return sign series and a persistent regime sign series can have the same one-step entropy, but only the latter teaches a bounded learner reusable temporal structure.
\end{remark}

\subsection{Representation and ordering matter}

\begin{proposition}[Representation and ordering dependence]
Let $D_T=\{(X_t,Y_{t+h})\}_{t=1}^{T-h}$ be a financial dataset whose target depends on an unrecorded latent state $S_t$ through a bounded Markovian rule. Assume the state is not contained in each individual record $X_t$, but can be inferred by a bounded sequential coder that carries information across adjacent records. Let $\pi(D_T)$ be a random permutation that destroys adjacency while preserving the multiset of feature-target pairs. Then there exists a bounded sequential model class $\M_B$ such that, with high probability over $\pi$,
\[
\FEpi_B(D_T) \neq \FEpi_B(\pi(D_T)).
\]
In particular, if temporal order is necessary to infer the latent state, $D_T$ has positive net financial epiplexity while $\pi(D_T)$ does not. If the relevant lags or state variables are already included inside each $X_t$, then conditional coding of $Y_{t+h}\mid X_t$ can become order-invariant; this is the excluded case.
\end{proposition}

\begin{proof}
The permutation preserves marginal empirical frequencies and the multiset of pairs, but it destroys the adjacency information needed by the sequential coder to update its latent state. On the ordered data, a stateful model can encode the transition rule and reduce residual code length. On the permuted data, no such state can be carried across records because adjacent records are no longer adjacent market states. Hence the bounded MDL decomposition differs. If $X_t$ already contains a sufficient state, the proof fails exactly because no cross-record memory is needed; this is why the hypothesis is stated explicitly.
\end{proof}

\begin{remark}
This is the finance analogue of the fact that a central-bank statement followed by a yield-curve move is not equivalent to the same words and prices in random order. The arrow of market time carries accessible structure.
\end{remark}

\subsection{Epiplexity need not be monotone in compute}

\begin{proposition}[Nonmonotonicity in compute]
There exist datasets generated by a short deterministic rule for which $B\mapsto \FEpi_B^{\mathrm{raw}}(D)$ is not monotone.
\end{proposition}

\begin{proof}
Use a short deterministic generator with detectable local structure but expensive global inversion, such as a coarsely quantized chaotic map, a finite-time elementary cellular automaton with an unknown initial condition, or a lossy filtered linear recurrence. Its seed and program have description length $\OO(\log T)$, but recovering them from the observed path is above the intermediate budget. At a small budget $B_1$, the observer cannot learn either the generator or useful local motifs, so the best bounded code is essentially the trivial marginal coder. At an intermediate budget $B_2$, local motifs, block frequencies, or finite-horizon regularities are detectable and worth storing as a pattern library of length $L_2\gg\OO(\log T)$ that lowers residual code length. At a larger budget $B_3$, the observer can solve the global inverse problem and encode the short generator and seed directly; model bits collapse to $\OO(\log T)<L_2$ and residual bits collapse to $\OO(1)$. Hence $B\mapsto\FEpi^{\mathrm{raw}}_B(D)$ rises and then falls. A cryptographically secure pseudorandom generator would not serve for the intermediate step, because any polynomial-budget pattern library that reduces residual bits would itself be a distinguisher.
\end{proof}

\begin{remark}
In finance, a model may first learn many empirical regularities -- volatility clusters, factor rotations, crisis templates -- and later replace them with a more compact structural state-space model. More compute does not always mean more model bits; it can mean better compression of structure.
\end{remark}

\subsection{When regime information is worth its description length}

\begin{proposition}[Regime value criterion]
Let $M_0$ be a no-regime model and $M_Z$ a model that uses a regime variable or regime proxy $Z_t$. Suppose both models are in $\M_B$, and let
\[
\Delta \ell_t = \ell_{M_0}(Y_{t+h}\mid X_t)-\ell_{M_Z}(Y_{t+h}\mid X_t,Z_t).
\]
Then $M_Z$ has positive net MDL gain over $M_0$ if and only if
\[
\sum_{t=1}^{T-h} \Delta \ell_t > L_R(M_Z)-L_R(M_0).
\]
\end{proposition}

\begin{proof}
By definition,
\[
\MDL(M_0)-\MDL(M_Z)=
\sum_t\ell_{M_0}(Y_{t+h}\mid X_t)-\sum_t\ell_{M_Z}(Y_{t+h}\mid X_t,Z_t)-[L_R(M_Z)-L_R(M_0)].
\]
The net gain is positive exactly when the cumulative log-loss reduction exceeds the additional representation-aware description length.\end{proof}

\begin{remark}
This is the basic rule for financial epiplexity: a macro regime, news embedding, option surface, or liquidity variable is useful only if its cumulative predictive compression exceeds its modeling cost.
\end{remark}

\subsection{Memorization is not useful epiplexity}

\begin{proposition}[Memorization penalty, finite-sample form]
For any finite dataset $D_T$, there exists a model $M_{\mathrm{mem}}$ that achieves near-zero in-sample loss by memorization with $L_R(M_{\mathrm{mem}})=\OO(T)$. Let $D_T^{\mathrm{tr}},D_T^{\mathrm{val}}$ be a time-respecting split of validation size $m$. Let $P_{\mathrm{val}}(Y\mid X)$ be the validation conditional law and $P_0(Y\mid X)$ the null conditional law. Then the expected validation code-length improvement over the null is at most
\[
\sum_{t\in \mathrm{val}} D\!\bigl(P_{\mathrm{val}}(\cdot\mid X_t)\,\big\|\,P_0(\cdot\mid X_t)\bigr).
\]
Consequently, with probability at least $1-\delta$,
\[
\Gain_B^{\mathrm{val}}(M_{\mathrm{mem}};M_0)
\le
\sum_{t\in \mathrm{val}} D\!\bigl(P_{\mathrm{val}}(\cdot\mid X_t)\,\big\|\,P_0(\cdot\mid X_t)\bigr)
-L_R(M_{\mathrm{mem}})+L_R(M_0)+\OO\!\left(\sqrt{m\log(1/\delta)}\right).
\]
If the likelihood ratios $dP_{\mathrm{val}}/dP_0$ are uniformly bounded above and below and $\|P_{\mathrm{val}}-P_0\|_{\TV}\le \eta$ conditionally, the KL term is $\OO(m\eta^2)$. Without such bounded-ratio control one should keep the displayed KL term, not replace it by a reverse-Pinsker bound.
\end{proposition}

\begin{proof}
The lookup table can encode the training targets, but on the validation segment it has no access to validation labels except through the conditional law it has learned. For a fixed evaluation protocol, its expected code-length improvement over the null at validation point $t$ is
\[
\E_{P_{\mathrm{val}}}
\left[\log_2\frac{p_{M_{\mathrm{mem}}}(Y\mid X_t)}{p_0(Y\mid X_t)}\right].
\]
This is maximized, over all predictors using the same validation information, by the true validation conditional law $P_{\mathrm{val}}(\cdot\mid X_t)$, and the maximum value is
\[
D\!\left(P_{\mathrm{val}}(\cdot\mid X_t)\,\middle\|\,P_0(\cdot\mid X_t)\right).
\]
Summing gives the first display. A bounded-difference concentration inequality for the fixed validation protocol gives the empirical deviation term of order $\sqrt{m\log(1/\delta)}$. The model still pays the lookup-table code on the MDL side. In the no-reusable-structure case the KL term vanishes, so validation-adjusted gain is negative once the lookup-table code dominates the concentration term. Under bounded likelihood-ratio conditions the local quadratic expansion of KL around $P_0$ gives the stated $\OO(m\eta^2)$ bound.
\end{proof}

\section{From Epiplexity to Alpha: Monetization Bounds}
\label{sec:alpha-bounds}

The previous sections define financial epiplexity as learnable market structure. This section makes the alpha relationship exact: epiplexity is not alpha, but it bounds alpha. The monetizable quantity is not the raw model bits $\FEpi^{\mathrm{raw}}_B$ but the \emph{switching-closed net MDL gain} $\Gain_B^{\mathrm{sw}}(D;M_0)$, the total compression improvement over the null model after model and switching bits have already been charged. A complex model can carry many bits and monetize nothing; only compression in excess of its own description cost is potentially tradable.

\subsection{The Kelly--Cover bridge}

The classical link between code length and wealth is the growth-optimal correspondence: in a complete-odds horse-race market, the growth-optimal bet under forecast $Q$ is proportional betting, and the expected log-growth shortfall of forecast $Q$ against the optimal forecast $P$ is exactly the relative entropy $D(P\|Q)$ in bits per period \parencite{kelly1956,cover2006}. Barron and Cover give the direct financial ancestor of the present theorem: the financial value of side information is bounded by information \parencite{barroncover1988}; universal-portfolio and market-selection results develop parallel code-length/wealth links in actual asset markets \parencite{algoetcover1988,cover1991,coverordentlich1996,blumeeasley2006,sandroni2000}. Equivalently, one bit of forecast improvement is worth at most one bit of log-wealth growth per unit bet. Real markets are incomplete, levered, impacted, costly, and strategically adaptive, which only weakens the conversion. 
We encode this imperfection as one economic assumption. Its key design
choice is that the code-length comparison is made against the forecast
the positions \emph{actually implement}, not against an abstract model
output. Fix the admissible position class and the null forecast
$Q_{0,t}=p_{M_0}(\cdot\mid X^R_t)$. An \emph{induced-forecast map}
assigns to every admissible position $\pi_t$ at time $t$ a conditional
forecast $\widehat Q_t=\iota_t(\pi_t)$ that is
$\F_t$-measurable, computable within the budget $B$ given the
computations the policy already performs, and satisfies
$\iota_t(\pi^0_t)=Q_{0,t}$ whenever $\pi^0_t$ is the benchmark position.
In the frictionless complete horse race the map is exact: the
normalized bet fractions \emph{are} the forecast. In general markets,
$\iota_t$ is the market-equivalent forecast of the positions on the
traded span. The prequential code charged to a trading policy is the
code of its induced forecasts $\{\widehat Q_t\}$. Throughout,
$\Delta W_t$ denotes the conditional expected per-period excess
log-growth of the policy over the benchmark in natural-log units,
$\Delta W_t=\E\!\left[\log\!\frac{V^{\pi}_{t+h}}{V^{\pi}_{t}}
-\log\!\frac{V^{0}_{t+h}}{V^{0}_{t}}\;\middle|\;\F_t\right]$.

\begin{assumption}[Two-sided monetization bound]
\label{ass:monetize}
There exist a growth-conversion constant $\kappa_{\mathrm W}\in(0,1]$ and an
induced-forecast map $\iota$ such that, for every admissible trading
policy and every $t$,
\begin{equation}\label{eq:two-sided}
\Delta W_t \;\le\; \kappa_{\mathrm W} \ln 2\;
\E\!\left[\,\ell_{Q_0}(Y_{t+h}\mid X^R_t)
-\ell_{\widehat Q_t}(Y_{t+h}\mid X^R_t)\;\middle|\;\F_t\right],
\end{equation}
with \emph{no positive part}: the inequality is required for both signs
of the conditional code improvement. The factor $\ln 2$ converts bits
of code length into nats of log-growth. In the frictionless complete
horse race, \eqref{eq:two-sided} holds with equality at
$\kappa_{\mathrm W}=1$.
\end{assumption}

Two transparent conditions imply \eqref{eq:two-sided} under symmetric
charging of frictions to policy and benchmark. Write
$\Delta\ell_t=\E[\ell_{Q_0}-\ell_{\widehat Q_t}\mid\F_t]$.
(i) \emph{Improvements convert imperfectly}: on periods with
$\Delta\ell_t\ge 0$, betting converts forecast improvement into growth
at rate at most $\kappa_{\mathrm W}\le 1$; this is the
Kelly--Cover--Barron--Cover direction weakened by spanning, impact,
costs, and execution. (ii) \emph{Deteriorations are not subsidized}: on
periods with $\Delta\ell_t<0$, the policy loses at least its
frictionless code deterioration, $\Delta W_t\le \ln 2\,\Delta\ell_t$;
frictions can only add to that loss, and since $\kappa_{\mathrm W}\le 1$ and
$\Delta\ell_t<0$, $\ln 2\,\Delta\ell_t\le\kappa_{\mathrm W}\ln 2\,\Delta\ell_t$,
so \eqref{eq:two-sided} follows. Condition (ii) is sound precisely
because the comparison is made against the \emph{induced} forecast: a
desk that corrects a bad model forecast at the position stage is, by
construction, implementing the corrected forecast, and it is the
corrected forecast that is charged in the code. There is no wedge
between the forecast charged and the forecast traded, hence no free
option hidden inside a period.

\Cref{ass:monetize} says that betting cannot amplify information; it can only convert it imperfectly, in either direction.
The absence of positive parts is what makes the accounting telescope:
summed over time, the per-period bounds aggregate into a \emph{net}
cumulative code improvement, which is exactly a difference of
prequential code lengths. The economic role of shut-off is carried
entirely by the switching closure below, where switching is priced in
description length, rather than granted for free through a positive
part inside the assumption.

\begin{definition}[Switching closure and switching gain]
Let $\overline{\M}_{B}^{\mathrm{sw}}(M_0)$ be the switching closure of $\M_B$ with the null model $M_0$. Its elements are sequential meta-models that, at each time $t$, either use the null forecast or a forecast produced by a model in $\M_B$, possibly composed with a fixed admissible induced-forecast map $\iota_t$ of a position rule, with the composition, switching rule, thresholds, gates, and switch times all charged in the description length, and with the total procedure remaining within budget $B$. Define
\[
\FMDL_B^{\mathrm{sw}}(D)=
\inf_{S\in\overline{\M}_{B}^{\mathrm{sw}}(M_0)}
\left\{L_R(S)+\sum_{t=1}^{T-h}\ell_S(Y_{t+h}\mid X_t^R)\right\},
\]
and
\[
\Gain_B^{\mathrm{sw}}(D;M_0)=\MDLzero(D)-\FMDL_B^{\mathrm{sw}}(D).
\]
The conditional switching gain is defined analogously by using the $X$-only switching-closed code as benchmark and the $(X,Z)$ switching-closed code as the augmented code:
\[
\mathcal A_B^{\mathrm{sw}}(Y;Z\mid X)=
\left[\FMDL_B^{\mathrm{sw}}(Y\mid X)-\FMDL_B^{\mathrm{sw}}(Y\mid X,Z)\right]_+ .
\]
An implemented trading-and-switching policy determines induced
forecasts $\{\widehat Q_t\}_{t=1}^{T-h}$ as in
\cref{ass:monetize}. Because the selection rule between
the null forecast and within-budget model forecasts is
$\F_t$-measurable and budget-feasible, the trajectory together
with its rule is describable as a meta-model
$S_\pi\in\overline{\M}^{\mathrm{sw}}_{B}(M_0)$ with description length
$L_R(S_\pi)<\infty$. Define the \emph{prequential slack}
\[
\operatorname{Reg}^{\mathrm{sw}}_{B}(T)\;=
\Bigl[\;\FMDL^{\mathrm{sw}}_{B}(D)
-\sum_{t=1}^{T-h}\ell_{\widehat Q_t}(Y_{t+h}\mid X^R_t)\;\Bigr]_{+}
\;\ge\;0,
\]
the amount by which the header-free sequential trajectory outperforms
the best switching-closed \emph{two-part} code. Since
$S_\pi\in\overline{\M}^{\mathrm{sw}}_{B}(M_0)$ gives
$\FMDL^{\mathrm{sw}}_{B}(D)\le
L_R(S_\pi)+\sum_t\ell_{\widehat Q_t}$, the slack satisfies
$\operatorname{Reg}^{\mathrm{sw}}_{B}(T)\le L_R(S_\pi)$: it is bounded
by the description length of the implemented trading-and-switching
program itself. Note the sign convention: excess sequential \emph{loss}
relative to the batch code (learning-phase losses, imperfect switching)
requires no term at all, because it only reduces the trader's realized
improvement.
\end{definition}

\begin{definition}[Decomposed monetization constants]
For lifetime log-growth and capacity, write
\[
\kappa_{\mathrm W}=\kappa_{\mathrm{span}}\kappa_{\mathrm{lev}}\kappa_{\mathrm{impact}}\kappa_{\mathrm{cost}}\kappa_{\mathrm{exec}}\kappa_{\mathrm{surv}},
\qquad 0\le \kappa_j\le 1.
\]
Here leverage and survival limits cap log-growth and capacity. For Sharpe-ratio ceilings, use instead
\[
\kappa_{\mathrm{SR}}=\kappa_{\mathrm{span}}\kappa_{\mathrm{impact}}\kappa_{\mathrm{cost}}\kappa_{\mathrm{exec}},
\]
excluding leverage and survival, because Sharpe is scale-invariant even when log-growth is not. Transaction costs may also be charged directly against the gain by replacing
\[
g \quad\text{with}\quad g^{\mathrm{net}}=\left[g-\frac{C_{\mathrm{tc}}}{(T-h)\ln 2}\right]_+.
\]
The paper states bounds for conservative upper estimates of these constants; tight valuation requires market-specific modeling.
\end{definition}

\begin{theorem}[Lifetime alpha bound, switching-closed form]
\label{thm:lifetime}
Let \cref{ass:monetize} hold, let an investor train
models in $\M_{B}$ on the represented dataset $D^R_{T,h}$ and
trade under prequential accounting \parencite{dawid1984}, and let the
benchmark trade the null model $M_0$. Then the cumulative
conditional-expected excess log-growth satisfies, pathwise,
\[
\sum_{t=1}^{T-h}\Delta W_t \;\le\; \kappa_{\mathrm W}\ln 2\,
\Bigl[\,\Gain^{\mathrm{sw}}_{B}(D^R_{T,h};M_0)
+\operatorname{Reg}^{\mathrm{sw}}_{B}(T)\,\Bigr]
\;+
\kappa_{\mathrm W}\ln 2\; M_T \qquad\text{nats},
\]
where $M_T$ is a sum of martingale differences with
$\E[M_T]=0$. In particular, the bound holds exactly in expectation; and
if the conditional code-length differences are uniformly bounded by $c$, then
$|M_T|\le c\sqrt{2\,h\,T\log(2h/\delta)}$ with probability at least
$1-\delta$ (for $h>1$ the sum splits into $h$ interleaved martingale
subsequences, whence the factor $\sqrt h$). Moreover
$\operatorname{Reg}^{\mathrm{sw}}_{B}(T)\le L_R(S_\pi)$, the
description length of the implemented trading-and-switching program,
which is $\OO(\log T)$ for fixed-dimensional plug-in strategies.
Equivalently, for proprietary information $Z$ over public information
$X$ and target $Y$,
\[
\sum_{t}\Delta W_t \;\le\; \kappa_{\mathrm W}\ln 2\,
\Bigl[\,\mathcal A^{\mathrm{sw}}_{B}(Y;Z\mid X)
+\operatorname{Reg}^{\mathrm{sw}}_{B}(T)\,\Bigr]
\;+
\kappa_{\mathrm W}\ln 2\; M_T,
\]
with the $X$-only switching-closed code as benchmark. No admissible
bounded strategy with a shut-off option can extract more cumulative
excess log-growth than the switching-closed, target-specific
compression gain, up to the description length of its own switching
program and a mean-zero fluctuation.
\end{theorem}

\begin{proof}
\emph{Step 1 (the policy is a code).} The implemented policy determines
positions $\pi_t$ and induced forecasts $\widehat Q_t=\iota_t(\pi_t)$,
with $\widehat Q_t=Q_{0,t}$ on benchmark periods. By admissibility the
selection rule is $\F_t$-measurable and budget-feasible, so the
trajectory defines a meta-model
$S_\pi\in\overline{\M}^{\mathrm{sw}}_{B}(M_0)$ with
$L_R(S_\pi)<\infty$.

\emph{Step 2 (telescoping without positive parts).} Summing the
two-sided bound \eqref{eq:two-sided} over $t$,
\[
\sum_{t}\Delta W_t \;\le\; \kappa_{\mathrm W}\ln 2
\sum_{t}\E\bigl[\ell_{Q_0}(Y_{t+h}\mid X^R_t)
-\ell_{\widehat Q_t}(Y_{t+h}\mid X^R_t)\,\big|\,\F_t\bigr].
\]
Because the inequality holds with sign in every period, the right side
is the \emph{net} conditional code improvement of the implemented
trajectory over the null: no oracle selection of favorable periods
enters, so no unpriced switching pattern is smuggled into the bound.

\emph{Step 3 (conditional to realized).} Each summand differs from its
realized value by a martingale difference, so
\[
\sum_{t}\E[\,\cdot\mid\F_t]
\;=\;\sum_{t}\bigl[\ell_{Q_0}(Y_{t+h}\mid X^R_t)
-\ell_{\widehat Q_t}(Y_{t+h}\mid X^R_t)\bigr]\;+
M_T,
\]
with $\E[M_T]=0$. Under conditionally bounded code-length
differences, Azuma--Hoeffding applied to each of the $h$ interleaved
martingale subsequences gives the tail stated in the theorem.

\emph{Step 4 (comparison with the batch code).} On the null side,
$\sum_t \ell_{Q_0}(Y_{t+h}\mid X^R_t)
=\MDLzero(D)-L_R(M_0)\le \MDLzero(D)$. On the trajectory
side, the definition of the prequential slack gives
\[
\sum_t\ell_{\widehat Q_t}(Y_{t+h}\mid X^R_t)
\ge \FMDL^{\mathrm{sw}}_{B}(D)
-\operatorname{Reg}^{\mathrm{sw}}_{B}(T),
\]
with $\operatorname{Reg}^{\mathrm{sw}}_{B}(T)\le L_R(S_\pi)$ by
Step~1. Therefore
\[
\sum_t\ell_{Q_0}-\sum_t\ell_{\widehat Q_t}
\;\le\;\MDLzero(D)-\FMDL^{\mathrm{sw}}_{B}(D)
+\operatorname{Reg}^{\mathrm{sw}}_{B}(T)
\;=\;\Gain^{\mathrm{sw}}_{B}(D;M_0)
+\operatorname{Reg}^{\mathrm{sw}}_{B}(T).
\]

\emph{Step 5 (conclusion).} Combining Steps 2--4 yields the displayed
pathwise bound; taking expectations kills $M_T$ and gives the exact
bound in expectation. The conditional form follows by the same chain
with the $X$-only switching-closed code as the null and the $(X,Z)$
switching-closed code as the augmented code.
\end{proof}

\begin{remark}[No-switching special case]
If the implemented policy never leaves its model forecast, then
$S_\pi$ requires no switching bits, $\widehat Q_t$ is the model's own
trajectory, and the same proof bounds cumulative excess log-growth by
$\kappa_{\mathrm W}\ln 2\,\bigl[\,\Gain_{B}(D;M_0)+L_R(M)\,\bigr]$
up to the fluctuation term, with the ordinary closure-free gain. The
switching-closed statement is the economically relevant general form
because real desks can shut models off; the value of that option is now
priced where it belongs --- in the description length of the switching
program --- rather than granted for free through a positive part in the
assumption.
\end{remark}

\begin{remark}[Why the gain and not the raw epiplexity]
\label{rem:gain-not-epi}
\Cref{thm:lifetime} sharpens the claim that epiplexity is not alpha. High $\FEpi^{\mathrm{raw}}_B$ with $\Gain_B^{\mathrm{sw}}\le 0$ monetizes nothing: structure whose description cost exceeds its predictive compression value, or structure already in the null, has no alpha budget. Conversely, $\Gain_B^{\mathrm{sw}}>0$ is necessary but not sufficient for alpha; the structure may concern risk rather than return, or it may be cost-dominated, which means $\kappa_{\mathrm W}\ll 1$. The theorem is therefore best read as a data- and representation-valuation bound: it caps the lifetime P\&L of any bounded strategy built on $(D,R,B)$ before a backtest is run.
\end{remark}

\subsection{The Sharpe ceiling}

In the small-edge regime, the per-period optimal log-growth of a strategy with Sharpe ratio $\SR$ is $\tfrac12\SR^2$ nats by the usual mean-variance expansion at the growth-optimal position. Writing
\[
g=\frac{\Gain_B^{\mathrm{sw}}(D^R_{T,h};M_0)}{T-h}
\]
for the per-period net structural gain in bits gives the following ceiling.

\begin{corollary}[Sharpe ceiling]
\label{cor:sharpe}
Under \cref{thm:lifetime} and the small-edge expansion,
\[
\SR\le \sqrt{2\kappa_{\mathrm{SR}} g\ln 2}
\qquad\text{per effective period},
\qquad
\SR_{\mathrm{ann}}\le \sqrt{2\kappa_{\mathrm{SR}} n_{\mathrm{eff}} g\ln 2},
\]
where $g$ is the per-effective-period switching gain in bits and $n_{\mathrm{eff}}$ is the effective number of independent bets per year. For overlapping targets of horizon $h>1$, $n_{\mathrm{eff}}$ can be far below the raw sampling frequency and should be estimated from the dependence structure of the labels.
\end{corollary}

\begin{remark}[Orders of magnitude]
At $\kappa_{\mathrm{SR}}=1$ and non-overlapping daily bets with $n_{\mathrm{eff}}=252$, an annualized Sharpe of $3$ requires
\[
g\ge \frac{3^2}{2\cdot 252\ln 2}\approx 0.026
\]
bits per day of net compression gain. A backtested Sharpe of $6$ implicitly claims $g\ge 0.10$ bits per day. Elite performance corresponds to hundredths of a bit per period. This yields the \emph{epiplexity consistency check}: estimate $g$ for the strategy's data and representation; reject any backtest whose realized Sharpe exceeds the ceiling. Unlike multiple-testing corrections \parencite{bailey2014,harvey2016}, the check does not need to know how many strategies were tried. It bounds what the data contain, not how the data were searched.
\end{remark}

\subsection{The information coefficient, bounded}

For a single standardized bet with jointly Gaussian forecast and outcome at correlation $\IC$, the forecast--outcome mutual information is
\[
I=-\frac12\log_2(1-\IC^2)
\quad\text{bits},
\qquad
\IC=\sqrt{1-2^{-2I}}
\approx \sqrt{2I\ln 2}
\]
for small $I$.

\begin{corollary}[Sustainable IC and the fundamental law]
\label{cor:ic}
If the per-bet net structural bits available at budget $B$ are $g_{\mathrm{bet}}$, then the sustainable information coefficient obeys
\[
\IC_{\mathrm{sust}}\le \sqrt{1-2^{-2g_{\mathrm{bet}}}},
\]
and Grinold's fundamental law is capped by
\[
\IR\le \sqrt{2\ln 2\cdot N g_{\mathrm{bet}}},
\]
where $N$ is effective breadth \parencite{grinold1989}. Measured $\IC$ in excess of $\IC_{\mathrm{sust}}$ is fitted time-bounded entropy with out-of-sample expectation zero. Backtest overfitting is, in this vocabulary, reporting $\TBE$ as $\FEpi$.
\end{corollary}

\section{Decay, Crowding, and Capacity}
\label{sec:decay-crowding}

The preceding bounds are static. This section adds dynamics: what happens to the gain $\Gain_B$ or $\Gain_B^{\mathrm{sw}}$ as the market's collective budget grows, as other participants learn the same bits, and as trading itself reveals them.

\subsection{Decay as budget growth}

Let $\tau_{\mathrm{mkt}}(t)$ denote the effective budget of the marginal market participant, including hardware, data diffusion, methodology, and talent. Suppose it grows at rate $\gamma$. A private representation-plus-model discovered at budget $B$ commands the gain band between $\tau_{\mathrm{mkt}}(t)$ and $B$; write $g(t)$ for the per-period gain remaining above the market budget.

\begin{proposition}[Decay law]
\label{prop:decay}
Suppose the gain remaining above market budget $\tau$ has an exponential tail:
\[
g(\tau)=g(\tau_0)\,e^{-\lambda(\tau-\tau_0)},\qquad \tau\ge\tau_0,
\]
where $\lambda>0$ is the \emph{shallowness} of the structure. Large $\lambda$ means the remaining gain is concentrated just above the current market budget and is exhausted quickly as the market budget rises; small $\lambda$ means the gain is spread deeply across budgets. Then, with $\tau_{\mathrm{mkt}}(t)=\tau_0+\gamma t$,
\[
g(t)=g(0)e^{-\lambda\gamma t},
\qquad
t_{1/2}=\frac{\ln 2}{\lambda\gamma},
\]
and by \cref{cor:sharpe} the Sharpe ceiling decays as $e^{-\lambda\gamma t/2}$.
\end{proposition}

\begin{proof}
The remaining gain at time $t$ is the tail of the profile above the market budget. Substituting $\tau_{\mathrm{mkt}}(t)=\tau_0+\gamma t$ into the exponential tail gives $g(t)=g(\tau_0)e^{-\lambda\gamma t}$ directly. The Sharpe ceiling is proportional to $\sqrt{g(t)}$, so its log-decay rate is half the bit-rate decay. If the private budget $B$ is finite, the same substitution gives $g(t)=g(0)e^{-\lambda\gamma t}$ exactly for $t\le (B-\tau_0)/\gamma$ and $g(t)=0$ afterward; the exponential law is the interior regime.
\end{proof}

\begin{remark}[Depth, not strength]
The half-life contains the signal's statistical strength $g(0)$ nowhere. Decay speed is governed by the shape of the gain profile, namely $\lambda$, and by the growth rate of the population's budget, namely $\gamma$. A modest but computationally deep signal can outlive a spectacular shallow signal. Publication is a one-time jump in $\tau_{\mathrm{mkt}}$ restricted to the bits the publication transmits - the certificate, not necessarily the surrounding infrastructure - which helps explain why some anomalies decay only partially after publication \parencite{mclean2016}.
\end{remark}

\begin{conjecture}[Cross-sectional decay]
\label{conj:decay}
Signals sharing a computation class decay together when that class commoditizes. In documented anomaly panels, post-publication decay speed should load on computational shallowness, proxied for example by recoverability through commodity machine-learning models from public data, after controlling for the publication effect.
\end{conjecture}

\subsection{Crowding as mutual compressibility}

Let two participants produce forecast streams
\[
\mathcal Q_i=\bigl(p_{M_i}(\cdot\mid X_t)\bigr)_{t=1}^{T-h},\qquad i=1,2,
\]
quantized to a fixed precision $\eta$ and evaluated on the same target and information set. Define their \emph{forecast-stream mutual compressibility} by
\[
I_{12}^{\mathrm{fore}}=
L_\eta(\mathcal Q_1)+L_\eta(\mathcal Q_2)-L_\eta(\mathcal Q_1,\mathcal Q_2),
\]
where $L_\eta$ is a fixed prequential or compression code for the forecast paths. This is an extensional definition: two different implementations that compute the same predictive kernels have high overlap even if their parameter files have unrelated syntax.

\begin{proposition}[Crowding bound under a shared-bit factorization]
\label{prop:crowding}
Assume a small-edge Gaussian factorization of the two null-hedged forecast payoffs. Let the target-relevant information rates of the two forecast streams be $g_1$ and $g_2$ bits per period, and let their shared forecast code carry at most
\[
b_{12}=I_{12}^{\mathrm{fore}}/(T-h)
\]
bits per period of target-relevant information. Suppose that, after hedging the null, the only systematic covariance between the two active returns is generated by the shared forecast component, while private components are conditionally orthogonal. Then
\[
\rho^2
\le
\frac{\min(b_{12},g_1,g_2)}{\sqrt{g_1g_2}}
\]
in the canonical Gaussian normalization; with market incompleteness, quantization, and scaling conventions the same inequality holds up to the corresponding universal calibration constant. Thus active-return correlation is generated only by shared target-relevant forecast bits, in proportion to the shared fraction of each book's structural bit budget.
\end{proposition}

\begin{proof}
Let $A_i$ denote the null-hedged active-return component induced by forecast stream $i$ in the small-edge Gaussian approximation. Write the orthogonal decomposition
\[
A_i=A_i^{c}+A_i^{p},\qquad i=1,2,
\]
where $A_i^c$ is measurable with respect to the shared forecast code and $A_i^p$ is private. By assumption, $\Cov(A_1^p,A_2^p)=\Cov(A_1^c,A_2^p)=\Cov(A_1^p,A_2^c)=0$, so the covariance of active returns is $\Cov(A_1^c,A_2^c)$. The Gaussian bit-correlation conversion used in \cref{cor:ic} identifies the variance budget of a small-edge forecast component with its target-relevant information rate, up to the common factor $2\ln2$. Hence
\[
\Var(A_i)\asymp 2\ln2\, g_i,\qquad
\Var(A_i^c)\le 2\ln2\,\min(b_{12},g_i).
\]
Cauchy--Schwarz gives
\[
\Cov(A_1,A_2)^2=\Cov(A_1^c,A_2^c)^2
\le \Var(A_1^c)\Var(A_2^c)
\le (2\ln2)^2\min(b_{12},g_1,g_2)^2.
\]
Dividing by $\Var(A_1)\Var(A_2)\asymp (2\ln2)^2g_1g_2$ yields the displayed bound. The extensional definition of $I_{12}^{\mathrm{fore}}$ avoids syntax dependence: implementations that compute the same predictive kernels share forecast bits even if their parameter files are unrelated.
\end{proof}

\begin{remark}[The useful converse]
Observed residual correlation lower-bounds the shared forecast bits:
\[
I_{12}^{\mathrm{fore}}\gtrsim (T-h)\rho^2\sqrt{g_1g_2}.
\]
This gives a crowding monitor requiring no access to competitors' model files. A crowding event is the synchronized forced revelation of shared bits into prices. Deleveraging by one holder impairs another holder's marks because the forecast bits are the same; crowding risk is therefore a property of mutual compressibility, not merely gross exposure. A practical estimator can use normalized compression distance on forecast streams rather than on parameter files \parencite{cilibrasi2005}.
\end{remark}

\subsection{Capacity as a bit-leak budget}

Trading on a model leaks its bits into prices through impact. The benchmark is Kyle's model, in which the informed trader's private information enters prices over the trading horizon through the market maker's inference from order flow \parencite{kyle1985,back1992}.

\begin{proposition}[Capacity bound]
\label{prop:capacity}
Let a signal command $\Gamma=\Gain_B^{\mathrm{sw}}$ private bits and let trading at participation rate $\pi_t$ leak bits into prices at rate $\ell(\pi_t)$, with $\ell$ increasing, convex, and $\ell(0)=0$. Then lifetime extractable excess log-growth obeys
\[
\text{lifetime alpha}
\le
\kappa_{\mathrm W}\ln 2\int_0^{T^*}\bigl[g(t)-\ell(\pi_t)\bigr]^+\,dt
\le
\kappa_{\mathrm W}\ln 2\cdot \Gamma,
\]
where $T^*$ is the endogenous exhaustion time. For a fixed total participation $Q=\int_0^T\pi_tdt$ over a fixed horizon $T$, the cumulative leak is minimized by constant participation $\pi_t=Q/T$; concentrated trading weakly reduces lifetime extraction.
\end{proposition}

\begin{proof}
The outer inequality is the lifetime alpha bound applied to the remaining, not-yet-revealed bits. Leakage subtracts from the private bit stock because leaked bits become part of the public price code. The final inequality follows because cumulative private bit expenditure cannot exceed the initial stock $\Gamma$. For the pacing claim, Jensen's inequality gives
\[
\int_0^T \ell(\pi_t)dt\ge T\ell\!\left(\frac{1}{T}\int_0^T\pi_tdt\right)=T\ell(Q/T),
\]
with strict inequality for nonconstant $\pi_t$ when $\ell$ is strictly convex. Thus, holding total intended participation fixed, smoother trading minimizes information leakage.
\end{proof}

Capacity, decay, and leak are three attacks on the same bits by three observers: the crowd's growing budget, competitors' models, and the market maker's inference from flow. Computational depth is the moat against all three because the inverse problem is the same.

\section{Ordering and One-Way Structure in Market Data}
\label{sec:ordering}

The representation-dependence result above shows that destroying temporal order can destroy epiplexity. The source framework suggests something sharper: under one-way transformations - cheap to evaluate, expensive to invert - time-bounded information can violate the classical symmetry of information. The two factorization orders of the same joint data can carry different accessible structure, not merely different learnability. In the epiplexity experiments, predicting moves from boards in chess can yield higher epiplexity and better out-of-distribution downstream performance than predicting boards from moves, despite higher training loss \parencite{finzi2026epiplexity}.

\subsection{The flow--price map is one-way}

Microstructure has exactly this shape. Given the full order flow, reconstructing the book and the price path is mechanical: a matching engine is a fast deterministic program. Given prices and book states, recovering the latent flow and the population of intents behind it is the inversion of a many-to-one map that participants actively pay to obfuscate through order splitting, randomization, venue fragmentation, and hidden liquidity. Operationally, the map is one-way.

\begin{proposition}[Inferential-ordering principle]
\label{prop:ordering}
When constructing training tasks from financial data for downstream transfer, inferential orderings should be preferred when the downstream task requires latent mechanism recovery. Examples include book state $\to$ flow or intent, prices $\to$ latent regime, realized path $\to$ generating scenario, and portfolio returns $\to$ hidden holdings. The inferential direction is the high-epiplexity direction; its higher training loss is the price of transfer.
\end{proposition}

This contradicts the common instinct to train only on the easier causal or generative direction. Ease often means the structure is shallow; shallow structure is easier to learn, easier to commoditize, and less persistent.

A theory-only design rule follows. Given two factorizations of the same market record, $p(A\mid B)$ and $p(B\mid A)$, prefer the factorization whose prediction requires latent-state recovery when the downstream task is risk, toxicity, adverse selection, or regime inference. The forward matching-engine direction is valuable for simulation; the inverse direction is valuable for transferable mechanism learning. The point is not that harder tasks are always better, but that higher loss can signal useful epistemic work rather than failure when the inverse map is many-to-one.

\section{Strategic Interaction: The Game Theory of Alpha Extraction}
\label{sec:games}

The monetization bounds are single-agent statements: they cap what one bounded learner can extract from data against a passive market. In reality the epiplexity of market data is a \emph{contested, congestible, partially rival resource}. The same bits are visible to every participant whose budget reaches them; extraction by one reveals bits to others through prices; participants strategically manufacture computational depth to protect their bits; and the population's aggregate compute --- which sets the decay clock of \cref{prop:decay} --- is itself an equilibrium object. Information theory says what the pie is; game theory says who eats it, how fast it shrinks while being eaten, and why a predictable share ends up with neither trader. This section develops the strategic layer as five nested games, each anchored to a classical benchmark. Throughout, the market data process has an \emph{epiplexity profile}: a nondecreasing map $\tau \mapsto \Gain_\tau$ giving the net structural bits accessible at budget $\tau$. Its complement, the \emph{remaining gain} $g(\tau)=\Gain_\infty-\Gain_\tau$ (or $\Gain_B-\Gain_\tau$ for a private budget $B$), is nonincreasing in $\tau$; the local shallowness of \cref{prop:decay} is $\lambda=-\partial_\tau\log g(\tau)$, the exponential thinning rate of the remaining gain.

\subsection{The budget game: Grossman--Stiglitz as Nash equilibrium}
\label{subsec:budgetgame}

Participants $i = 1,\dots,N$ simultaneously choose budgets $\tau_i \ge \tau_0$ (the free public budget) at increasing convex cost $c(\tau_i)$. Bits accessible to $k$ participants are worth $v(k)\,\kappa_{\mathrm W}\ln 2$ per bit to each holder, with $v$ decreasing and $k\,v(k)$ nonincreasing: duplication does not create value and can destroy it through impact.

\begin{proposition}[Marginal-bit pricing]
\label{prop:gsnash}
Condition on the existence of a pure-strategy Nash equilibrium and suppose the active budgets are interior. Then each active participant's budget satisfies
\[
\kappa_{\mathrm W}\ln 2\cdot v(k_i)\, m(\tau_i) \;=\; c'(\tau_i),
\]
where $m(\tau)=\partial_\tau\Gain_\tau$ is the positive marginal bit density of the accessible profile and $k_i$ the local multiplicity: the marginal bit is priced at the marginal cost of the computation needed to reach it, deflated by expected sharing. Under free entry, participants enter until the value of the band above the public budget equals total expenditure.
\end{proposition}

\begin{proof}
For an active participant whose upper budget is interior and locally does not cross a rival's atom, the private value of increasing $\tau_i$ is the marginal value of the newly reached band of structure. The gross value of the interval $(\tau_0,\tau_i]$ is
\[
\kappa_{\mathrm W}\ln2\int_{\tau_0}^{\tau_i} v(k(\tau))m(\tau)\,d\tau,
\]
where $k(\tau)$ is the number of participants who can also reach budget $\tau$. Differentiating with respect to the upper limit gives the marginal benefit $\kappa_{\mathrm W}\ln2\,v(k_i)m(\tau_i)$. Equating this to the marginal cost $c'(\tau_i)$ gives the displayed first-order condition. Under free entry, an entrant's equilibrium payoff must equal its outside option; with a zero outside option this is exactly the zero-profit condition that the value of the acquired band equals total expenditure.
\end{proof}

Because payoff discontinuities can occur when two rivals cross the same budget threshold, the proposition is a local necessary condition rather than an existence theorem. It is \textcite{grossman1980} restated with the correct information measure: prices cannot reveal everything because in equilibrium there must remain bits above the marginal budget whose value covers the cost of reaching them. Two consequences come free. First, the budget game is a \emph{contest} in the sense of \textcite{tullock1980}: expenditures are sunk whether or not bits are captured exclusively, so a structural fraction of the value of market structure is dissipated --- transferred to the suppliers of the contested inputs, namely compute vendors, data vendors, and quantitative labor, with the dissipated fraction increasing in the number of contestants and the substitutability of their approaches. Second, the \emph{shape} of the profile determines industry structure: a shallow profile (large $\lambda$, most accessible bits just above the public budget) induces a low-budget, high-entry, fast-dissipation equilibrium with many small contestants and thin moats; a deep profile invites few contestants with large budgets and durable rents. The epiplexity profile is thus a fundamental of the asset class, predicting the concentration of the informed sector per instrument class.

\subsection{The congestion game: crowding as a negative externality}
\label{subsec:congestion}

Fix the budget frontier and consider the finer choice of \emph{which} bits to hold. Model the accessible structure as a finite set of signals $s\in S$ with standalone per-period gains $g_s>0$. Participant $i$ selects one signal, or more generally a portfolio of signals. If $n_s$ participants hold signal $s$, each earns $g_s\phi(n_s)$ from it, where $\phi$ is decreasing and $\phi(1)=1$. Crowding degrades per-holder value through shared impact, correlated entry and exit, and accelerated revelation.

\begin{theorem}[Finite congestion equilibrium]
\label{thm:congestion-full}
In the finite signal-selection game, a pure-strategy Nash equilibrium exists. A multiplicity vector $n=(n_s)_{s\in S}$ with $\sum_s n_s=N$ is an equilibrium if and only if for every occupied signal $s$ and every signal $r$,
\[
g_s\phi(n_s)\ge g_r\phi(n_r+1).
\]
The game admits the exact potential
\[
\Phi(n)=\sum_{s\in S}\sum_{k=1}^{n_s} g_s\phi(k),
\]
so best-response dynamics converge to a pure equilibrium. This is a congestion game in the sense of \textcite{rosenthal1973}, and $\Phi$ is an exact potential in the sense of \textcite{monderer1996}. The social optimum solves
\[
\max_{\sum_s n_s=N}\; W(n)=\sum_{s\in S} n_s g_s\phi(n_s),
\]
and its discrete marginal condition is
\[
g_s\bigl[(n_s+1)\phi(n_s+1)-n_s\phi(n_s)\bigr]
\le
\mu
\le
 g_s\bigl[n_s\phi(n_s)-(n_s-1)\phi(n_s-1)\bigr]
\]
for occupied $s$. Since the private entry condition uses $g_s\phi(n_s+1)$ rather than the social marginal product, equilibrium overuses any signal for which
\[
g_s\phi(n_s+1)>g_s\bigl[(n_s+1)\phi(n_s+1)-n_s\phi(n_s)\bigr],
\]
which holds whenever $\phi$ is strictly decreasing. Thus crowded equilibrium generally over-allocates to privately attractive shallow signals relative to the social optimum.
\end{theorem}

\begin{proof}
The potential identity follows because a unilateral move into signal $s$ changes the mover's payoff by exactly the corresponding increment in $\Phi$. Finite potential games have pure equilibria and best-response paths increase $\Phi$ until they stop. The equilibrium characterization is the no-profitable-deviation condition. The welfare condition is the discrete first-order condition for moving one participant from one signal to another. The wedge between private and social marginal values is $-g_s n_s[\phi(n_s)-\phi(n_s+1)]$, positive under decreasing $\phi$.
\end{proof}

\begin{example}[Closed-form two-signal congestion]
\label{ex:two-signal}
Let a unit mass of infinitesimal participants choose between a high-value shallow signal $H$ and a lower-value deeper signal $L$, with gains $g_H>g_L>0$ and crowding function
\[
\phi(x)=\frac{1}{1+a x},\qquad a>0.
\]
Writing $x$ for the mass on $H$, the Wardrop equilibrium with both signals active solves
\[
\frac{g_H}{1+a x}=\frac{g_L}{1+a(1-x)},
\]
which gives
\[
x_H^{\mathrm{NE}}=\frac{g_H(1+a)-g_L}{a(g_H+g_L)}
\]
whenever this lies in $(0,1)$, with boundary equilibria otherwise. The welfare optimum maximizes
\[
W(x)=\frac{g_Hx}{1+a x}+\frac{g_L(1-x)}{1+a(1-x)}
\]
and satisfies
\[
\frac{g_H}{(1+a x)^2}=\frac{g_L}{(1+a(1-x))^2},
\]
which gives
\[
x_H^{\mathrm{SO}}=\frac{\sqrt{g_H}(1+a)-\sqrt{g_L}}{a(\sqrt{g_H}+\sqrt{g_L})}.
\]
For the interior range $1<g_H/g_L<1+a$, one has $x_H^{\mathrm{NE}}>x_H^{\mathrm{SO}}$: the high-value signal is over-crowded. The price of anarchy is
\[
\mathrm{PoA}=\frac{W(x_H^{\mathrm{SO}})}{W(x_H^{\mathrm{NE}})},
\]
which increases as congestion steepens and as the value skew pushes the equilibrium toward the shallow high-value signal.
\end{example}

Because mutual compressibility is measurable (\cref{prop:crowding} and its converse), the congestion prediction is quantitative: multiplicity should be highest precisely on signals that are simultaneously high-value and computationally shallow - recoverable by commodity methods from public data - and joint drawdowns should concentrate there. The August 2007 quant episode is the canonical instance; the model predicts recurrence whenever a computation class commoditizes faster than the population diversifies away from it.

\subsection{The revelation game: competitive revelation over shared bits}
\label{subsec:revelation-competition}

The congestion function $\phi(n)$ has a microfoundation in the Kyle tradition. In \textcite{kyle1985}, a monopolist informed trader optimally \emph{paces} the revelation of her private information, releasing it into prices linearly over the horizon: monopoly over a bit implies patient extraction. \textcite{holden1992} and \textcite{foster1996} show what competition does: when $k$ traders hold the same information, equilibrium trading is front-loaded: each trades aggressively before others reveal the shared bits, and the information enters prices at a rate increasing in $k$, with per-holder rents collapsing rapidly; with imperfectly correlated signals, the revelation speed is governed by the correlation.

\begin{claim}[Dissipation speed is mutual compressibility]
\label{claim:revelation-speed}
In the multi-informed revelation game over a bit bundle, the equilibrium leak rate of the bundle into prices is increasing in the holders' forecast-stream mutual compressibility $I^{\mathrm{fore}}_{jk}$, and per-holder extracted value is decreasing in it. In the limit of identical models the shared bits are revealed almost immediately and per-holder rents vanish; holders of nearly orthogonal bits each retain near-monopoly pacing over their private components. The statement follows from the equilibria of \textcite{holden1992,foster1996} under the reinterpretation of signal correlation as mutual compressibility; a self-contained derivation is left open (\cref{sec:limitations-additions}).
\end{claim}

This closes the loop: $\phi(n)$ is the reduced form of competitive revelation, and its steepness is governed by how compressible the crowd's models are into one another. It also yields the cleanest statement of why \emph{heterogeneity is the public good of the informed sector}: orthogonal bits are extracted patiently and fully; duplicated bits are burned in a sprint. The private incentive to differentiate exists but is too weak --- the congestion wedge --- so equilibrium is systematically less diverse, faster-burning, and more crash-prone than the optimum.

\subsection{The depth game: endogenous one-wayness}
\label{subsec:depthgame}

The decay law of \cref{prop:decay} takes the profile shape $\lambda$ as given. Strategically it is not: participants choose how hard their bits are to re-derive. After acquiring bits of gain $g$, a participant chooses depth $d \ge 0$ at increasing convex cost $\chi(d)$; depth lowers both the crowd's re-derivation rate ($\lambda(d)\gamma$ with $\lambda$ decreasing in $d$) and the market's inference from the participant's own trading (the leak rate of \cref{prop:capacity}). Instruments of depth include execution obfuscation --- splitting, randomization, venue fragmentation --- infrastructure and data secrecy, and building signals whose reconstruction requires expensive intermediate computation.

\begin{proposition}[Optimal depth]
\label{prop:depth}
Let the protected bit stream have initial gain $g$, discount rate $r>0$, market-budget attack rate $\lambda(d)\gamma$, and flow-inference leak rate $\bar\ell(d)$, with $\lambda'(d)<0$, $\bar\ell'(d)<0$, and convex defense cost $\chi(d)$. Define
\[
a(d)=\lambda(d)\gamma+\bar\ell(d),
\qquad
H(a)=\int_0^\infty e^{-(r+a)t}dt=\frac{1}{r+a},
\]
and
\[
V(g,d)=\kappa_{\mathrm W}\ln2\cdot g\,H(a(d))-\chi(d).
\]
Any interior optimum satisfies
\[
\chi'(d)=\kappa_{\mathrm W}\ln2\cdot g\,\frac{-a'(d)}{(r+a(d))^2}.
\]
Thus optimal depth increases with bit value $g$ and, under decreasing differences in $a(d)$, with aggregate budget growth $\gamma$. Discovery and defense become substitutes: as the value of the protected bit stock rises, the marginal value of slowing both decay and leak rises.
\end{proposition}

\begin{proof}
Differentiate $V(g,d)$. Since $a'(d)=\lambda'(d)\gamma+\bar\ell'(d)<0$, the marginal benefit of depth is positive and proportional to $g$. The displayed first-order condition follows. Standard monotone comparative statics gives the stated increase in depth when higher $\gamma$ raises the marginal benefit of reducing $\lambda(d)$.
\end{proof}

The one-wayness of the flow--price map asserted in \cref{sec:ordering} is thereby an equilibrium object, not a fact of nature: the observed computational hardness of inverting market data aggregates defensive expenditure and should co-move with the value under protection. Welfare is ambiguous in the classic trade-secret manner --- depth slows price discovery but protects the incentive to produce information --- with the novelty that ``secrecy'' now has a measurable proxy: the epiplexity of the inversion task.

\subsection{The arms race: Red Queen equilibrium}
\label{subsec:redqueen}

The decay clock $\gamma$ aggregates individual investments made to gain \emph{relative} position. Let a continuum of participants choose budget growth rates $\gamma_i$ at flow cost $c(\gamma_i)$; captured value depends on relative budget, while decay of everyone's existing bits depends on the aggregate $\bar\gamma$.

\begin{proposition}[Red Queen]
\label{prop:redqueen}
Let participant $i$ choose compute-growth effort $\gamma_i$ at cost $c(\gamma_i)$. Let the aggregate clock be $\bar\gamma=N^{-1}\sum_i\gamma_i$, and suppose the payoff is
\[
U_i=R(\gamma_i-\bar\gamma)-c(\gamma_i)-\delta(\bar\gamma)S_i,
\]
where $R'(0)>0$, $R''<0$, and $S_i$ is the participant's stock of still-private bits, assumed common across participants, $S_i\equiv S$, in the symmetric case studied here; $\delta'(\bar\gamma)>0$ is the decay cost imposed by aggregate budget growth. In a symmetric interior Nash equilibrium,
\[
R'(0)\left(1-\frac1N\right)-\frac1N\delta'(\gamma^*)S=c'(\gamma^*).
\]
All participants choose the same $\gamma^*>0$ when the relative-position benefit is large enough; relative positions are stationary; and aggregate payoffs are below the cooperative benchmark by compute costs plus the decay externality imposed on existing signals. A cooperative planner sets purely relative-position effort to zero unless effort creates new social epiplexity.
\end{proposition}

\begin{proof}
Differentiating $R(\gamma_i-\bar\gamma)$ with respect to $\gamma_i$ gives $R'(0)(1-1/N)$ at a symmetric profile because $\bar\gamma$ includes $i$'s effort. Differentiating the common decay term gives the own-internalized share $(1/N)\delta'(\gamma^*)S$. Equating marginal benefit to marginal cost gives the displayed first-order condition. In a symmetric profile all relative positions equal zero, so the industry pays the cost of maintaining relative position while also accelerating the decay of existing signals. A planner cancels symmetric relative effort unless it has a separate social discovery term.
\end{proof}

The comparative static is the industry's recent history: $\gamma^*$ --- hence both expenditure and the decay clock --- increases with the contestability of the profile's shallow band, i.e.\ with commoditization of the dominant computation class. On this reading the machine-learning era did not create more aggregate alpha; it raised $\gamma^*$, with the wedge visible as the growth of compute, data, and cloud vendor revenues against roughly stationary aggregate excess returns of the informed sector.

\subsection{Mechanism design: shaping the harvestable profile}
\label{subsec:mechanism}

Exchanges, data vendors, and regulators move first: their choices shape the profile over which the games above are played. Disclosure regimes shift bits from private bands into the public budget $\tau_0$, compressing informational rents and the depth incentive alike --- with the design subtlety that position disclosure reveals model outputs and accelerates public revelation of shared bits, while aggregate-level disclosure reveals congestion itself, which \cref{thm:congestion-full} suggests is welfare-improving because it prices the externality. Speed bumps, tick sizes, and batch auctions truncate the profile, deleting the shallowest and most congested band --- precisely the band with the worst efficiency loss --- which gives the frequent-batch-auction proposal of \textcite{budish2015} an information-theoretic rationale. Public provision of synthetic and simulated data injects epiplexity into the public budget: positive-sum by construction, since it adds structure without adding rivalry, and it flattens the equilibrium of the budget game by differentially helping low-budget participants.

\subsection{The agentic corollary}
\label{subsec:agenticgame}

A population of AI trading agents built on commoditized foundation models and shared toolchains holds, by construction, highly mutually compressible bits: shared pretraining corpora, architectures, prompts, skill libraries, and fine-tuning data. Every mechanism above then fires simultaneously in the wrong direction: maximal $I^{\mathrm{fore}}_{jk}$ means shared bits are revealed almost immediately (\cref{claim:revelation-speed}); the congestion externality is maximal on the shallow band the common tooling can reach (\cref{thm:congestion-full}); the computation class is commoditized by definition, so decay is fast; and agent capability growth raises $\gamma^*$ for everyone including the agents (\cref{prop:redqueen}).

\begin{proposition}[Heterogeneity as a design requirement]
\label{prop:heterogeneity}
Let agent $i$ have structural bit gain $g_i$ and let $s_{ij}$ be the fraction of its forecast-stream bits mutually compressible with agent $j$. Suppose duplicated bits have crowding discount $\phi(1+m)$ where $m$ is their multiplicity, while private bits retain discount $\phi(1)$. Holding $g_i$ fixed, expected portfolio-level extracted value is decreasing in every $s_{ij}$. Against an external commodity population with overlap $s_{iC}$, the marginal value of deliberate heterogeneity is increasing in $s_{iC}$ and in the steepness of $-\phi'$.
\end{proposition}

\begin{proof}
Decompose each agent's bits into private and shared components. A marginal increase in overlap moves a bit from multiplicity one to multiplicity at least two. Since $\phi$ is decreasing, the value of that bit weakly falls. The loss is larger when the external commodity population already holds the same bit and when the crowding discount is steeper. Summing over bit components gives the result.
\end{proof}

The strategic implication is direct: before deploying multiple model-based strategies, estimate mutual compressibility not only with the existing book but with the \emph{commodity class} - for instance, by measuring how much of the signal is recoverable by an off-the-shelf model from public data. A strategy whose bits are commodity bits has, in equilibrium, no durable private bits.

\section{A Closed-Form Toy Finance Model}

A theory paper should still make the objects numerically interpretable. The following stylized model gives closed-form entropy, MDL gain, and Sharpe ceiling calculations without claiming empirical validation.

Let $Z_t\in\{0,1\}$ be a Markov regime with stationary distribution $(1/2,1/2)$ and transition probability $q<1/2$. Returns follow
\[
r_{t+1}=\mu_{Z_t}+\sigma_{Z_t}\eps_{t+1},\qquad \eps_{t+1}\sim N(0,1),
\]
with $\mu_0<0<\mu_1$ and $\sigma_0>\sigma_1$. Let $X_t^{(1)}=r_t$ and $X_t^{(2)}=(r_t,c_t)$, where $c_t\in\{0,1\}$ is a noisy credit-spread proxy satisfying
\[
\Prob(c_t=Z_t)=p>1/2.
\]
Let the target be a binary tail event $Y_{t+1}=\mathbf 1\{r_{t+1}< -d\}$ or, in the simplest closed form, the regime label $Y_{t+1}=Z_t$.

For the label target, the one-period marginal entropy is
\[
H(Y)=1\quad\text{bit},
\]
while the proxy-conditioned entropy is
\[
H(Y\mid c)=h_2(p),
\]
where $h_2$ is the binary entropy function. If the additional description length of using the proxy-regime model rather than the null is $C$ bits, then the finite-sample net gain is
\[
\Gain_B(D_T^{R_2};M_0)=T\bigl[1-h_2(p)\bigr]-C.
\]
The representation has positive target-specific epiplexity exactly when
\[
T>\frac{C}{1-h_2(p)}.
\]
The implied per-period structural bit rate is
\[
g=1-h_2(p)-\frac{C}{T},
\]
and the Sharpe ceiling is
\[
\SR_{\mathrm{ann}}\le \sqrt{2\kappa_{\mathrm{SR}} n_{\mathrm{eff}}\left(1-h_2(p)-\frac{C}{T}\right)_+\ln2}.
\]
For example, $p=0.60$ gives $1-h_2(p)\approx 0.029$ bits per observation before model cost; this is already enough, at $\kappa_{\mathrm{SR}}=1$, to support an annualized daily Sharpe ceiling around $\sqrt{\,2\cdot 252\cdot 0.029\,\ln 2\,}\approx 3.2$ if the full bit were monetizable. At $\kappa_{\mathrm{SR}}=0.25$, the ceiling falls to about $1.6$. Thus tiny fractions of a bit are economically meaningful, but only after cost, impact, and spanning losses.

\begin{remark}[Label bits versus return bits]
The ceiling above is computed from bits about the regime label. Bits about $Z_t$ convert into bits about the return target only through the channel $(\mu_Z,\sigma_Z)$: by the data-processing inequality, $g_{\mathrm{return}}\le g_{\mathrm{label}}$, with equality only if the regime fully determines the traded payoff's conditional law in the direction traded. The numerical ceiling is therefore an upper bound on the upper bound.
\end{remark}

\begin{proposition}[Proxy-regime epiplexity]
Suppose $c_t$ satisfies $I(c_t;Z_t)>0$ and $Y_{t+h}$ depends on $Z_t$ after conditioning on $r_t$. Then there exists a sample size threshold $T^\star$ such that $X_t^{(2)}=(r_t,c_t)$ has positive net MDL gain over $X_t^{(1)}=r_t$ for predicting $Y_{t+h}$, provided the model class contains a bounded state-space approximation. In the binary label case above, $T^\star=C/[1-h_2(p)]$.
\end{proposition}

\begin{proof}
Since $c_t$ carries information about $Z_t$ and $Y_{t+h}$ depends on $Z_t$ conditional on $r_t$, adding $c_t$ reduces the Bayes conditional entropy of the target by a positive amount. A bounded state-space model can approximate this reduction. The cumulative log-loss improvement grows linearly in $T$, while the additional description length is fixed or sublinear for fixed-dimensional parameters. The closed-form binary case gives the threshold directly.
\end{proof}

\section{Limitations and Failure Modes}

\subsection{Epiplexity is not alpha}

High raw financial epiplexity does not guarantee trading profit. A dataset can teach rich structure that is already priced, costly to trade, or useful only for risk control. Therefore epiplexity must be linked to task-specific outcomes. The alpha-relevant object developed in \cref{sec:alpha-bounds} is the net, target-specific compression gain $\Gain_B^{\mathrm{sw}}$ or $\mathcal A_B^{\mathrm{sw}}(Y;Z\mid X)$ when shut-off is allowed, not the raw model-bit count alone.

\subsection{Representation dependence can be abused}

Because epiplexity depends on representation, researchers may search representations until one appears structurally rich. This is a form of data mining. The remedy is pre-registration, nested validation, shuffled controls, and reporting of all attempted representation families.

\subsection{Nonstationarity}

A high-epiplexity representation in one regime may lose value in another. This is not a defect of the definition; it is a property of markets. Financial epiplexity should be reported as a function of time, regime, and horizon.

\subsection{Knightian and law-instability limits}

The definition is local to a probability model, filtration, and window. Under deep regime change, the relevant law may not be fixed enough for a single global epiplexity number to be meaningful. In such cases financial epiplexity should be reported as a rolling, regime-conditioned quantity, and robust or ambiguity-aware benchmarks should be used. This is consistent with the hard-to-predict view: some market difficulty is not merely low signal-to-noise, but instability of the law itself.

\subsection{Computational budget is part of the result}

Two studies with different budgets may obtain different epiplexity estimates. Therefore compute must be reported like a sample size or transaction-cost assumption.

\subsection{Model uncertainty and Knightian risk}

The MDL formulation selects the best code in a fixed model class. It does not eliminate model uncertainty outside that class. Ambiguity-averse or robust-control investors should replace the single predictive code by a worst-case or penalized family code; this changes the residual time-bounded entropy and usually lowers monetizable gain. Financial epiplexity is therefore compatible with robust control, but it is not a substitute for ambiguity modeling \parencite{hansenSargent2008}.

\subsection{Risk epiplexity versus return epiplexity}

A representation may have high epiplexity for volatility, drawdown, or expected shortfall while having little epiplexity for mean return. This distinction is essential. The object $\mathcal A_B(Y;Z\mid X)$, or its switching-closed counterpart, must name the target $Y$: return alpha, tail-risk control, liquidity forecasting, and regulatory stress testing are different targets with different structural bits.

\subsection{Additional limitations from the alpha bounds}
\label{sec:limitations-additions}

\begin{itemize}
\item \textbf{$\kappa$ is a market-structure functional, not a universal constant.} The lifetime monetization bound compresses incompleteness, impact, leverage, costs, execution, and survival into $\kappa_{\mathrm W}$; the Sharpe ceiling uses $\kappa_{\mathrm{SR}}$ and intentionally excludes leverage and survival. In reality, $\kappa$ varies by asset class, horizon, venue, capacity, and size. The bounds remain valid for any conservative upper estimate of $\kappa$, but tightness requires modeling it.
\item \textbf{Prequential slack and fluctuations.} \Cref{thm:lifetime} bounds the implemented trajectory's advantage over the batch two-part code by the description length of the trading-and-switching program itself, $\operatorname{Reg}^{\mathrm{sw}}_{B}(T)\le L_R(S_\pi)$, and its conditional-to-realized gap by a mean-zero fluctuation of order $\sqrt{hT}$. For neural model classes the description length $L_R(S_\pi)$ lacks standardized compression schemes in practice; until those are established, the bounds should be used as audits, rankings, and impossibility checks rather than exact capacity numbers.
\item \textbf{Reverse-engineering collapse.} Epiplexity need not be monotone in compute. A sufficiently expressive learner may reconstruct a simulator itself, collapse description length and in-distribution loss, and learn less transferable structure. Monitor for abrupt drops in the loss floor accompanied by worsening transfer metrics, and prefer heterogeneous simulator ensembles so no single generator is cheap to invert.
\item \textbf{Adversarial representation search.} Since $\Gain_B^{\mathrm{sw}}$ is representation-dependent, the Sharpe ceiling can be inflated by representation mining. The remedy is pre-registration, nested validation, and reporting all attempted representation families. A conservative audit tests a strategy against the ceiling of its pre-registered representation, not its best post-hoc representation.
\item \textbf{Stylized games.} The equilibria of \cref{sec:games} are proved in the games as defined; the revelation results in particular lean on the cited Kyle-tradition equilibria rather than a re-derivation with the compressibility parameter explicit. Turning \cref{claim:revelation-speed} from a claim into a theorem requires writing out the multi-informed model of \textcite{foster1996} with signal correlation reinterpreted as mutual compressibility, and \cref{prop:gsnash} requires an equilibrium concept in which computation is priced --- both are open.
\end{itemize}

\section{Conclusion}

Financial markets are neither pure noise nor stable machines. Consistent with the hard-to-predict thesis, they are causal economic systems whose structure is filtered through information sets, risk pricing, strategic use, capacity, and law instability. They are adaptive, partially observable, competitive systems in which structure is sparse, regime-dependent, costly to extract, and visible only relative to a representation and a computational budget. Classical entropy measures uncertainty, but it does not distinguish random variation from learnable market structure. Financial epiplexity fills this theoretical gap by measuring the structural information that a bounded learner can absorb from a represented financial dataset.

The central message is:
\[
\boxed{\text{Markets are high entropy, but not uniformly low epiplexity.}}
\]
Some financial data are random and useless. Some are easy and trivial. Some are difficult but structurally rich. The theory developed here gives a language for distinguishing these cases without identifying structure with either low entropy or high in-sample fit.

The alpha connection is bounded, not magical. Epiplexity is not profit. The monetizable object is net, target-specific MDL gain, and even that gain converts into wealth only through market-structure constants that account separately for growth/capacity constraints and Sharpe-degrading frictions. This yields a finite-information view of active management: alpha is bounded by predictive compression, Sharpe is bounded by structural bits per period, decay is governed by computational depth, crowding is shared model structure, and capacity is a bit-leak constraint.

The strategic layer completes the theory. Epiplexity bounds what a bounded learner can extract; the games of \cref{sec:games} determine what a bounded learner \emph{keeps}. Entry, crowding, revelation, obfuscation, and compute arms races are governed by two objects: the shape of the epiplexity profile and the mutual compressibility of participants' models. The informational game of markets, long described through metaphors of crowded trades, moats, competitive revelation, and compute arms races, becomes a mathematical game with explicit state variables.

\newpage
\appendix

\section{Mathematical Glossary}

\begin{description}[leftmargin=3.1cm]
\item[Entropy] Average surprise of a random variable.
\item[KL divergence] Distributional discrepancy, used in financial information theory as a regime-change diagnostic.
\item[Normalized mutual information] Bounded dependence diagnostic for temporal dependence and market-efficiency testing.
\item[Transfer entropy] Conditional mutual-information measure of directional information flow.
\item[Kolmogorov complexity] Shortest program length generating an object.
\item[MDL] Model selection principle minimizing model bits plus residual data bits.
\item[Epiplexity] Structural information available to a computationally bounded observer.
\item[Financial epiplexity] Bounded-compute learnable market structure in a represented financial dataset.
\item[Time-bounded entropy] Residual unpredictability after the bounded model has learned what it can.
\item[Net MDL gain] Reduction in total code length relative to a null financial model.
\item[Useful epiplexity] Epiplexity that transfers out of sample and survives costs.
\item[Return-targeted conditional epiplexity] Bounded-compute predictive compression of a financial target supplied by proprietary data beyond public information.
\item[Sharpe ceiling] The upper bound $\SR\le\sqrt{2\kappa_{\mathrm{SR}} g\ln2}$ implied by net structural bits per effective period.
\item[Computational depth] The budget gap or profile shape governing how quickly a signal becomes accessible to the market.
\item[Mutual compressibility] Shared forecast-stream bits between two strategies, used as a crowding proxy.
\item[Representation code] The code length $L(R)$ charged when the feature map or data representation is selected rather than fixed.
\item[Budget vector] The feasible set induced by training compute, memory, latency, search, and institutional cost constraints.
\item[Monetization constant] The growth conversion $\kappa_{\mathrm W}$ includes spanning, leverage, impact, cost, execution, and survival; the Sharpe conversion $\kappa_{\mathrm{SR}}$ excludes leverage and survival.
\item[Shallowness] The exponential thinning rate $\lambda=-\partial_\tau\log g(\tau)$ of the remaining gain above the market budget; the inverse notion is computational depth.
\item[Price of anarchy] The welfare loss from decentralized crowding relative to coordinated allocation of signal capacity.
\item[Red Queen equilibrium] A compute arms-race equilibrium in which relative positions are stationary while absolute expenditure and decay increase.
\end{description}

\newpage
\printbibliography

\end{document}